\documentclass[12pt, journal,onecolumn]{IEEEtran}

\usepackage{algorithmic}

\usepackage{graphicx}
\usepackage{mathptmx}
\usepackage{amsthm}
\usepackage{color}
\usepackage{balance}
\newcommand{\newreptheorem}[2]{%
\newenvironment{rep#1}[1]{%
 \def\rep@title{#2 \ref{##1}}%
 \begin{rep@thm}}%
 {\end{rep@thm}}}

\newreptheorem{theorem}{Theorem}
\newtheorem*{rep@thm}{\rep@title}

\newcommand{\F}{\mathbb{F}}
\newcommand{\E}{\mathbb{E}}
\newcommand{\mS}{\mathcal{S}}

\usepackage[utf8]{inputenc} 
\usepackage[T1]{fontenc}
\usepackage{url}
\usepackage{ifthen}
\usepackage{nicefrac, xfrac}
\usepackage{newtxtext}
\usepackage{newtxmath}

\usepackage[hidelinks]{hyperref}
\usepackage{cleveref}
\usepackage{epsfig, cite}
\usepackage{amsfonts}
\usepackage{graphicx}
\usepackage{soul}
\usepackage{latexsym}
\usepackage{colortbl}
\usepackage{comment}
\usepackage[dvipsnames]{xcolor}
\usepackage{caption}
\usepackage{subcaption}

\DeclareMathOperator{\rank}{rank}

\DeclareMathOperator{\Ima}{Im}

\newcommand\qbin[3]{\left[\begin{smallmatrix} #1 \\ #2 \end{smallmatrix}\right]_{#3}}

\newtheorem{thm}{Theorem}[section]
\newtheorem{lem}[thm]{Lemma}

\newtheorem{cor}{Corollary}

\newtheorem{defn}{Definition}[section]

\newtheorem{rem}{Remark}
\newtheorem{cla}{Claim}
\newtheorem{prob}{Problem}

\begin{document}
	
\title{Optimal Codes for the Coverage Depth Problem and the Performance of Random Codes} 
\author{
\IEEEauthorblockN{
\textbf{R. Gross}\IEEEauthorrefmark{1},
\textbf{Y. Grunbaum}\IEEEauthorrefmark{1},
\textbf{M. Bertuzzo}\IEEEauthorrefmark{2},
\textbf{E. Yaakobi}\IEEEauthorrefmark{1},
\textbf{A. Ravagnani}\IEEEauthorrefmark{2}
}

\IEEEauthorblockA{
\IEEEauthorrefmark{1}
Department of Computer Science,\\
Technion---Israel Institute of Technology,\\
Haifa 3200003, Israel\\
\{roee.g, yitzchakg, yaakobi\}@cs.technion.ac.il
}

\IEEEauthorblockA{
\IEEEauthorrefmark{2}
Department of Mathematics and Computer Science,\\
Eindhoven University of Technology,\\
 the Netherlands\\
\{m.bertuzzo, a.ravagnani\}@tue.nl
}


}

\IEEEoverridecommandlockouts

\maketitle

\begin{abstract}
DNA storage systems retrieve information by randomly sampling synthesized DNA strands, making the number of reads required for successful recovery a fundamental performance measure. This motivates the coverage depth problem: for given code parameters, determine a linear code that minimizes the expected number of randomly sampled columns required to span the entire information space. While MDS codes are known to be optimal whenever they exist, identifying optimal codes in parameter regimes where MDS codes do not exist remains largely open. In this work, we derive a general formula for the expectation in the coverage depth setting, and apply this result to establish the optimality of the $q$-ary simplex code for the parameters that allow its existence. Moreover, we show the optimality of the $q$-ary Hamming code by relating coverage depth to independent sets in the dual code and exploiting log-concavity properties of matroid basis-generating polynomials. We further analyze random linear codes, derive an exact expression for their expected coverage depth, and show that, in the constant-rate regime, their additive gap from the MDS benchmark remains bounded independently of the code length.
\end{abstract}

\section{Introduction}
DNA-based storage systems represent a promising solution for addressing the unprecedented growth of digital information that needs to be stored globally~\cite{rydning2022worldwide}.
Extreme compactness (orders of magnitude denser than any traditional storage medium), durability under suitable conditions, and low maintenance costs~\cite{alliance2021preserving, markowitz2023biology} are distinctive characteristics that make DNA-based systems particularly attractive for long-term data storage.

The DNA storage pipeline involves multiple steps: After encoding data into sequences over the DNA alphabet $\{A,C,G,T\}$ and dividing them into blocks, these blocks are synthesized into artificial DNA molecules known as \emph{strands}. Current limitations restrict the length of the synthesized strands~\cite{leproust2010synthesis}; furthermore, copies of each strand are produced that may contain errors. The resulting molecules are then stored in a container in an unordered manner. The process required to retrieve the stored information is sequencing, during which the strands are translated into digital sequences, called \emph{reads}, which are noisy copies of the stored strands. The reads are then decoded back into the original information.
A key feature of this process is that the DNA sequences are \emph{randomly} read to decode the original information.

Although various works have demonstrated the potential of DNA as a storage medium~\cite{tabatabaei2015rewritable, blawat2016forward, bornholt2016dna, yazdi2017portable, organick2018random, anavy2019data, barlev2021deep}, critical issues related to the efficiency of DNA sequencers~\cite{alliance2021preserving, yazdi2015dna, shomorony2022information} remain, resulting in slow throughput and high costs compared to other storage technologies. This bottleneck is closely related to the notion of \emph{coverage depth}~\cite{heckel2019characterization}, defined as the ratio between the number of reads that are sequenced and the number of designed DNA strands. Thus, an opportunity to improve efficiency and reduce costs consists of reducing the coverage depth by reducing the number of reads needed for successful recovery.

In a recent paper~\cite{barlev2025coveryourbases}, various problems aimed at enhancing the reliability and effectiveness of data retrieval were introduced and analyzed for the first time, with special emphasis on the \emph{random access coverage depth problem}; in this model, the user's data is represented by $k$ information strands, which are encoded into $n$ strands through a matrix $G \in \F_q^{k \times n}$, typically referred to as a generator matrix of an error-correcting code. The goal is to compute, and possibly minimize, the expected number of columns to draw to recover a specific information strand among the $k$ that have been stored. A similar version is represented by the \emph{DNA coverage depth problem}, where all the $k$ information strands must be retrieved.

Both problems have been studied primarily under the assumptions of error-free synthesis and sequencing, and uniform distribution when drawing the columns of the matrix.
The random access model has attracted growing interest in different directions, such as identifying the structural properties of generator matrices that are advantageous~\cite{gruica2025combinatorial}, providing optimal constructions for specific cases~\cite{gruica2026geometryofcodes, boruchovsky2025randomaccess}, and extending this model at the file level~\cite{abraham2024covering}.
Furthermore,~\cite{bodur2026coverage, wang2026randomaccess} improved upon these constructions, deriving bounds and computing asymptotic performances, while~\cite{hofmeister2026randomaccess} focused on the optimization problem in an asymptotic regime in the binary case.

On the other hand, the DNA coverage depth problem was addressed in~\cite{barlev2025coveryourbases}, where the authors proved that MDS codes achieve the best result in minimizing the expected number of reads by drawing connections to the classical coupon collector's problem~\cite{feller1957introduction, newman1960double, erdos1961classical, flajolet1992birthday}. Moreover, in~\cite{cao2025optimizing}, real DNA storage experimental data were used to extend the problem to noisy channels with log-normal distribution. Building on the framework presented in~\cite{barlev2025coveryourbases}, the DNA coverage depth problem was studied in~\cite{bertuzzo2025coverage, bertuzzo2026coverage} in regimes where MDS codes do not exist. More specifically, closed-form expressions were derived for the expected value of columns for $q$-ary simplex codes, Hamming codes, first-order Reed-Muller codes, and also for ternary Golay codes. 

In the first part of this paper, we address an open problem presented in~\cite{bertuzzo2026coverage}: Our goal is to identify codes that minimize the expected number of reads required to retrieve all the information (i.e., optimal codes) for specific choices of parameters, none of which admits the existence of MDS codes. More specifically, we establish the optimality of various code families: $q$-ary simplex codes, first-order $q$-ary Reed-Muller codes, and $q$-ary Hamming codes.

Concurrently with the preparation of this manuscript, Kong et al.~\cite{kong2026coveragedepthproblemdistributed} independently established the same optimality result for the $q$-ary simplex code in the coverage depth problem.

The second part is inspired by a recent paper by Wang and Yaakobi~\cite{wang2026randomaccess}, where random matrices were used to obtain bounds on the optimal coverage depth, showing in particular that, at constant rate, its additive gap from the MDS bound remains bounded. In this work, we study the random-code ensemble itself and derive an exact expression for its average coverage depth, which enables us to tighten the upper bound on the additive gap from the MDS coverage time by a factor of two compared to the bound established in \cite{wang2026randomaccess}. While \cite{wang2026randomaccess} established that the gap is linear in $n$ within the constant redundancy regime, we demonstrate that our approach yields the same result.

The remainder of this paper is organized as follows. In Section~\ref{section:Preliminaries}, we introduce the notation and definitions used throughout the paper. Section~\ref{section:GeneralFormula} presents a general formula for the expectation and an application of it to the first-order $q$-ary Reed-Muller codes.
Section~\ref{section:SimplexCode}, Section~\ref{section:ReedMuller}, and Section~\ref{section:HammingCode} establish the optimality of the \(q\)-ary Simplex code, the first-order $q$-ary Reed-Muller code, and the \(q\)-ary Hamming code, respectively. Finally, in Section~\ref{section:Random Codes}, we analyze the coverage depth of random linear codes.

\section{Preliminaries}\label{section:Preliminaries}

Throughout the paper, $k$ and $n$ are positive integers with $2 \leq k \leq n$, \(q\) denotes a prime power, \(\mathbb{F}_q\) denotes the finite field of order \(q\), and \([n]=\{1,\ldots,n\}\). We denote the \(n\)-th harmonic number by \(H_n=\sum_{j=1}^{n}1/j\). For a matrix \(G\in\mathbb{F}_q^{k\times n}\), we denote its columns by \(g_1,\ldots,g_n\). For \(S\subseteq[n]\), the matrix \(G_S\) is the submatrix of \(G\) formed by the columns indexed by \(S\). We use \(\operatorname{span}(G_S)\) for the subspace spanned by these columns and \(\operatorname{rank}(G_S)\) for its dimension.

In a DNA-based storage system, the data is stored as a vector of length $k$ whose entries are themselves vectors (called \emph{strands}) of length $\ell$ over the DNA alphabet $\Sigma = \{A,C,G,T\}$.
We embed $\Sigma^\ell$ into a finite field $\mathbb{F}_q$ and use a linear block code $\mathcal{C} \subseteq \mathbb{F}_q^n$ of dimension $k$ to encode an information vector $\smash{(x_1, \dots, x_k) \in \mathbb{F}_q^k}$ to an encoded vector $\smash{(y_1, \dots, y_n) \in \mathbb{F}_q^n}$. 

Since the $k$ information strands are encoded using a generator matrix~$G \in \mathbb{F}_q^{k \times n}$, there is a one-to-one correspondence between the encoded strands and the columns of~$G$; therefore, recovering the $i$-th information strand is equivalent to recovering the $i$-th standard basis vector, that is, it must be in the span of the already recovered columns of $G$, since we can see these columns as vectors in $\mathbb{F}_q^k$.

As already mentioned, a full-rank matrix \(G\in\mathbb{F}_q^{k\times n}\) whose row space is \(\mathcal{C}\) is called a generator matrix of \(\mathcal{C}\), and we write \(\mathcal{C}_G\) for the code generated by \(G\). The dual code of \(\mathcal{C}\) is
\[
\mathcal{C}^{\perp}=\left\{x\in\mathbb{F}_q^n:\langle x,c\rangle=0\text{ for every }c\in\mathcal{C}\right\}.
\]
A generator matrix of \(\mathcal{C}^{\perp}\) is a parity-check matrix of \(\mathcal{C}\). Given a generator matrix \(G\), we use \(G^\perp\) to denote a generator matrix of \(\mathcal{C}_G^\perp\); in particular, \(\ker(G^\perp)=\mathcal{C}_G\).



We now define the coverage process studied in this paper. Let \(G\in\mathbb{F}_q^{k\times n}\) have rank \(k\), 
and suppose that the columns of $G$ are drawn uniformly at random with repetition. The coverage time of \(G\), denoted by $T_G$, is the minimum number of columns of $G$ needed to span $\F_q^k$,
and the coverage depth of \(G\) is its expected coverage time, namely
\[
\mathbb{E}[G]=\mathbb{E}[T_G].
\]

As an example, the coverage process of an MDS code terminates as soon as \(k\) distinct column indices have been sampled, since every $k$ columns of an MDS matrix are linearly independent.

The value of \(\mathbb{E}[G]\) is invariant under elementary row operations. It therefore depends only on the code generated by $G$. This means that the quantity 
$\E[\mathcal{C}]$ is well defined for a linear code $\mathcal{C} \subseteq \F_q^n$ as $\E[G]$, where $G$ is \textit{any} generator matrix of $\mathcal{C}$. We then use $\E[G]$ and $\E[\mathcal{C}]$ interchangeably.

A primary motivation for the study of coverage depth comes from DNA-based data storage.
During DNA readout, individual molecules are sampled from a pool, and direct random access to the encoded symbols is generally unavailable.
As a result, recovery requires collecting samples until they jointly provide sufficient information to reconstruct the stored data.
From a coding-theoretic perspective, it is therefore natural to seek codes that minimize the expected number of samples required for recovery.

The main problem addressed in this paper, which has been introduced in~\cite{bertuzzo2025coverage}, is defined below.
\begin{prob}[\textbf{The optimal coverage depth problem}] \label{problem_1}
    For given values of $n$, $k$ and $q$, compute the value
$$\mathbb{E}_{\textnormal{opt}}[n,k]_q \triangleq \min\{\mathbb{E}[\mathcal C] \, : \,  \mathcal C \subseteq \F_q^n \, \text{ is a code of dimension $k$}\},$$
    and construct
    a code $\mathcal C$ attaining the minimum. 
\end{prob}

The first fundamental result about Problem~\ref{problem_1} was obtained in~\cite{barlev2025coveryourbases}.

\begin{thm}[\text{see \cite[Corollary~2]{barlev2025coveryourbases}}] \label{MDS_bound_thm}
    For any code $\mathcal{C} \subseteq \mathbb{F}_q^n$ of dimension $k$, we have 
    \[
    \mathbb{E}[\mathcal{C}] \ge \sum_{i=0}^{k-1} \frac{n}{n-i} = n(H_n - H_{n-k}).
    \]
    Furthermore, the lower bound is attained with equality if and only if $\mathcal{C}$ is an MDS code. 
\end{thm}

\begin{rem}
    Theorem~\ref{MDS_bound_thm} provides a lower bound on the expectation and solves Problem~\ref{problem_1} for any choice of parameters $n, k$ and $q$ for which there exists an $[n,k]_q$ MDS code. In particular, assuming that the MDS conjecture holds, we can write
    \[
    \mathbb{E}_{\text{opt}}[n,k]_q = n(H_n - H_{n-k}) \text{ when } q \ge n-1.
    \]
\end{rem}

In~\cite{bertuzzo2026coverage}, the authors studied several code families and derived closed formulas for their expected number of reads. We mention the most important results for our purposes.

\begin{thm}[\text{see \cite[Theorem~3.1]{bertuzzo2026coverage}}]
    Let $\mathcal{C} \subseteq \mathbb{F}_q^n$ be the $q$-ary simplex code of dimension $k$, where $n = (q^k-1)/(q-1)$. We have 
    \begin{equation*} 
    \mathbb{E}[\mathcal{C}] = k + \sum_{i=1}^k \frac{q^{i-1}-1}{q^k - q^{i-1}}.
    \end{equation*}
\end{thm}

\begin{thm}[\text{see \cite[Theorem~5.1]{bertuzzo2026coverage}}]
    Let $\mathcal{C} \subseteq \mathbb{F}_q^n$ be the $q$-ary Hamming code of redundancy $r$, where $n=(q^r-1)/(q-1)$. We have 
    \[
    \mathbb{E}[\mathcal{C}] = nH_n -  \sum_{\ell=1}^{r} \frac{1}{\binom{n-1}{n-\ell}}\frac{\prod_{i = 0}^{\ell-1} \frac{q^r-q^i}{q-1}}{\ell !}.
    \]
\end{thm}

Furthermore, a conjecture on the optimality of the $q$-ary simplex code has also been formulated~\cite[Conjecture~3.2]{bertuzzo2026coverage}. In this work, we establish this conjecture and prove that the $q$-ary Hamming code is optimal as well.


\section{A General Formula for the Expectation} \label{section:GeneralFormula}
In this section, we give a general formula for the expected number of reads one needs to draw to recover the information that has been stored. This formula will then be applied in Section~\ref{section:SimplexCode} to show the optimality of simplex codes.

From now on, let $X$ be a space such that $\dim_{\F_q}(X) = k$, and let $p$ be a probability distribution defined on $X$. For any $x \in X$, let $p(x)$ denote the probability associated to the element $x$, and for any subspace $V \le X$, let
\[
p(V) = \sum_{x \in V} p(x).
\]
In other words, this is the total probability mass of the subspace $V$. Note that, for any pair of subspaces $U,V \le X$ with $V \le U$, we have $0 \le p(V) \le p(U) \le 1$. Furthermore, we assume that $p(V) < 1$ for any $V \lneq X$, while $p(X) =1$. These conditions are equivalent to requiring that the support of $p$ spans $X$.

The process of drawing columns to recover the standard basis vectors can be visualized as building chains of dimension-increasing subspaces of the form
\[
0 = V_0 < V_1 < \dots < V_{k-1} < V_k = X, 
\]
with $\dim(V_i) = i$ for $0 \le i \le k$. In particular, the step corresponding to the transition from the subspace $V_i$ to the subspace $V_{i+1}$, for $0 \le i \le k-1$, coincides with drawing the $(i+1)$-th linearly independent column.

Moreover, let
\[
\mS = \{ (V_0, V_1, \dots, V_k) \, : \, 0 = V_0 < V_1 < \dots < V_k = X, \, \dim(V_i) = i \}
\]
be the set of all possible configurations.
The following auxiliary lemma focuses on the probability of a single step within the chain.

\begin{lem} \label{lem:dimension_increasing}
    Let $V \le X$ be a subspace with $\dim(V) = i$, where $0 \le i \le k-1$, and let $U \ge V$ be a subspace with $\dim(U) = i + 1$. 
    Suppose that the drawn columns first span the subspace $V$.
    The probability that $U$ is the next span we obtain is
    \[
    \Pr[V_{i+1} = U \, | \, V_i = V] = \frac{p(U) - p(V)}{1 - p(V)}.
    \]
\end{lem}
\begin{proof}
    For $r \ge 1$, suppose that the $r$-th draw allows us to get $U$. This means that the previous $r-1$ draws were columns inside $V$. This probability is equal to
    \[
    p(V)^{r-1} (p(U)-p(V)).
    \]
    To conclude, we sum over all possible values of $r$, getting
    \[
    \begin{split}
        \Pr[V_{i+1} = U \, | \, V_i = V] &= \sum_{r=1}^\infty p(V)^{r-1} (p(U)-p(V)) \\
        &= (p(U) - p(V)) \sum_{r=0}^\infty p(V)^r = \frac{p(U) - p(V)}{1 - p(V)}.
    \end{split}
    \]
\end{proof}

As a consequence, we can derive the probability of obtaining a specific chain.

\begin{cor}
    Let $S = (V_0, V_1, \dots, V_k) \in \mS$ be a chain of dimension-increasing subspaces, where $V_i \le X$ for every $0 \le i \le k$. We have
    \[
    \Pr[S] = 
    \prod_{i=1}^k \frac{p(V_i) - p(V_{i-1})}{1 - p(V_{i-1})}.
    \]
\end{cor}

On the other hand, given a chain $S \in \mS$, let $n_i \ge 1$ be the waiting time, after the selected columns first span the $(i-1)$-dimensional space $V_{i-1}$, until they first span $V_i$, where $1 \le i \le k$. In other words, $n_i$ is the number of draws needed for the $i$-th step in the chain.
The probability of getting the chain $S$ with waiting times $(n_1, \dots, n_k)$ is then
\[
\Pr[S, (n_1, \dots, n_k)] = \prod_{i=1}^k p(V_{i-1})^{n_i-1} \left(p(V_i) - p(V_{i-1})\right).
\]

We now have all the ingredients for the main result of this section.

\begin{thm} \label{thm:general_formula} 
     Let $T$ be the random variable that counts the number of draws needed to span $X$. We have
     \[
     \E[T] = \sum_{S = (V_0, \dots, V_k) \in \mS} \Pr[S] \sum_{j=1}^k \frac{1}{1-p(V_{j-1})}.
     \]
\end{thm}
\begin{proof}
The expectation can be expressed by summing over all the possible configurations of $\mS$ and all the possible $k$-tuples of waiting times as 
\[
\begin{split}
    \E[T] &= \sum_{S \in \mS} \sum_{n_1, \dots, n_k \ge 1} (n_1 + \cdots + n_k) \Pr[S, (n_1, \dots, n_k)] \\
    &= \sum_{S \in \mS} \sum_{n_1, \dots, n_k \ge 1} (n_1 + \cdots + n_k) \prod_{i=1}^k p(V_{i-1})^{n_i-1} \left(p(V_i) - p(V_{i-1})\right).
\end{split}
\]
Note that the sums involving the waiting times are independent. Moreover, the $i$-th factor of the product depends only on $n_i$ and is nonnegative, since $0 \le p(V_{i-1}) \le p(V_i) \le 1$ for every $1 \le i \le k$.
We can then exchange the order of sums and products, obtaining
\[
\begin{split}
    \E[T] =& \sum_{S \in \mS} \sum_{j=1}^k \sum_{n_1= 1}^\infty \cdots \sum_{n_k=1}^\infty n_j \prod_{i=1}^k p(V_{i-1})^{n_i-1} \left(p(V_i) - p(V_{i-1})\right) \\
    =& \sum_{S \in \mS} \sum_{j=1}^k \left( \left(p(V_j) - p(V_{j-1})\right) \sum_{n_j=1}^\infty n_j p(V_{j-1})^{n_j-1} \right)\prod_{\substack{i=1 \\ i \neq j}}^k \left(p(V_i) - p(V_{i-1})\right) \sum_{n_i = 1}^\infty p(V_{i-1})^{n_i-1}.
\end{split}
\]
Since $p(V_i) < 1$ for all $0 \le i \le k-1$, we can apply the convergence properties of geometric series
\[
\sum_{n=1}^\infty x^{n-1} = \frac{1}{1-x} \quad \text{and} \quad \sum_{n=1}^\infty nx^{n-1} = \frac{1}{(1-x)^2}.
\]
Therefore,
\[
\begin{split}
    \E[T] &= \sum_{S \in \mS} \sum_{j=1}^k \frac{p(V_j) - p(V_{j-1})}{(1-p(V_{j-1}))^2} \prod_{\substack{i=1 \\ i \neq j}}^k \frac{p(V_i) - p(V_{i-1})}{1-p(V_{i-1})} \\
    &= \sum_{S \in \mS} \sum_{j=1}^k \frac{1}{1-p(V_{j-1})} \prod_{i=1}^k \frac{p(V_i) - p(V_{i-1})}{1-p(V_{i-1})} \\
    &= \sum_{S \in \mS} \sum_{j=1}^k \frac{\Pr[S]}{1-p(V_{j-1})} = \sum_{S \in \mS} \Pr[S] \sum_{j=1}^k \frac{1}{1-p(V_{j-1})}.
\end{split}
\]
\end{proof}

We now show that applying Theorem~\ref{thm:general_formula} to any first-order $q$-ary Reed-Muller code allows us to derive a significantly simpler formula for their expectation than the one presented in~\cite[Theorem 7.3]{bertuzzo2026coverage}. Recall that, for an integer $m \ge 1$, the first-order $q$-ary Reed-Muller code $\mathcal{C} \subseteq \F_q^n$ has length $n = q^m$ and dimension $k=m+1$; see~\cite[Definition~15]{jurrius2012weightenum}.

It is convenient to introduce some auxiliary quantities that will also be used in the subsequent sections. Given a generator matrix $G \in \F_q^{k \times n}$, the drawing procedure induces a probability distribution $p_G$ on $\mathbb{F}_q^k$ given by 
\[ p_G(x) = \frac{|\{i\in[n]:g_i=x\}|}{n}. \] 
For a subspace $V\leq\mathbb{F}_q^k$, let $m_G(V) = |\{i\in[n]:g_i\in V\setminus\{0\}\}|$, and let $z_G=|\{i\in[n]:g_i=0\}|$.
Then 
\[ p_G(0)=\frac{z_G}{n} \qquad\text{and}\qquad p_G(V)=\frac{m_G(V)+z_G}{n}. \] 


\begin{cor} 
    Let $\mathcal{C} \subseteq \F_q^n$ be the first-order $q$-ary Reed-Muller code of dimension $m+1$, where $n = q^m$.
    \[
    \E[\mathcal{C}] = m+1 + \sum_{j=1}^{m} \frac{1}{q^{j}-1}.
    \]
\end{cor}
\begin{proof}
A natural generator matrix for $\mathcal{C}$ is 
\[
G = 
\begin{pmatrix}
    1 & 1 & \cdots & 1 \\
    v_1 & v_2 & \cdots & v_{q^m} 
\end{pmatrix},
\]
where the $v_i \in \F_q^m$ are pairwise distinct. In other words, each vector of $\F_q^m$ appears exactly once. The columns of $G$ are all and only the points of the affine chart $\mathcal{A} = \{ (1,v) \, : \, v \in \F_q^m \}$. Projectively, these are the normalized representative of the points of $PG(m,q) \setminus PG(H)$, where $H = \{ (0,v) \, : \, v \in \F_q^m\}$.

Let $V \le F_q^{m+1}$ be a $j$-dimensional subspace. We want to count the quantity $m_{G}(V)$, which corresponds to the cardinality of $V \cap \mathcal{A}$. We distinguish two cases based on the containment relation with $H$. If $V \subseteq H$, then $m_{G}(V) = 0$. 

On the other hand, if $V \nsubseteq H$, consider the projection map $\pi_0 : \F_q^{m+1} \rightarrow \F_q$ onto the first coordinate, namely $\pi_0(v) = x_0$ for every vector $v = (x_0, \dots, x_m)$. We have $\ker(\pi_0) = H$ and, since $V \nsubseteq H$, the map restricted to $V$ is nonzero and, therefore, surjective. The rank-nullity theorem yields
\[
\dim(\ker(\pi_0\vert_V)) = \dim(V) - \dim(\Ima(\pi_0\vert_V)) = j-1,
\]
or, equivalently, $\dim(V \cap H) = j-1$. 
Let $v' \in V \cap \mathcal{A}$ (such a vector exists since the map $\pi_0\vert_V$ is surjective). It is not difficult to see that
\[
V \cap \mathcal{A} = v'+ (V \cap H),
\]
which means that $V \cap \mathcal{A}$ is an affine translation of $V \cap H$. Since $\dim(V \cap H) = j-1$, we have that 
\[
m_{G}(V) = |V \cap \mathcal{A}| = q^{j-1}.
\]
Therefore, 
\[
p_{G}(V) = \begin{cases}
0 &\text{if } V \subseteq H,\\
q^{j-1-m} &\text{otherwise}.
\end{cases}
\]
Moreover, since every drawn column has the first coordinate equal to $1$, every possible chain of $\mathcal{S}$ is of the form $(0, V_1, \dots, V_{m+1})$, with $V_j \nsubseteq H$ for all $j \ge 1$.
Thus, we have 
\[
p_{G}(V_0) = 0 \quad \mbox{and} \quad p_{G}(V_j) = q^{j-1-m} \quad \mbox{for $1 \le j \le m+1$.}
\]
Combining all of the above yields
\[
\begin{split}
    \E[\mathcal{C}] &= \sum_{S \in \mS} \Pr[S] \left( 1 + \sum_{j=1}^{m} \frac{1}{1-q^{j-1-m}} \right) \\
    &= 1 + \sum_{j=1}^{m} \frac{q^{m+1-j}}{q^{m+1-j}-1} = m+1 + \sum_{j=1}^{m} \frac{1}{q^j-1}.
\end{split}
\]
\end{proof}

\section{Optimality of the $q$-ary Simplex Code} \label{section:SimplexCode}
This section is devoted to the optimality of the $q$-ary simplex code, which was conjectured in~\cite{bertuzzo2025coverage}. Concurrent with the preparation of this manuscript, Kong et al.~\cite{kong2026coveragedepthproblemdistributed} independently presented a similarity proof of the optimality of the $q$-ary simplex code.

We first apply the formula obtained in Theorem~\ref{thm:general_formula} to establish a general lower bound on the expectation; moreover, we provide the cases where equality holds. This result will then be translated into the coding theory language to prove the optimality of the $q$-ary simplex code.

First, we introduce an additional notation: For a space $X$ of dimension $k$ over $\F_q$, we denote the projective space obtained from $X$ by $PG(k-1,q)$. For an integer \(m\geq1\), let \(\mathcal{P}_m\) denote the set of one-dimensional subspaces of \(\mathbb{F}_q^m\). Every nonzero vector belongs to a unique element of \(\mathcal{P}_m\), and
\[
|\mathcal{P}_m|=\frac{q^m-1}{q-1}.
\]

\begin{thm} \label{thm:simplex_optimality}
    Let $T$ be the random variable that counts the number of draws needed to span $X$.
    We have
    \begin{equation} \label{eq:simplex_lower_bound}
        \E[T] \ge \sum_{j=0}^{k-1} \frac{q^k-1}{q^k-q^j}.
    \end{equation}
    In particular, the bound is attained with equality if and only if the probability distribution on $X$ is such that $p(0) = 0$ and the induced distribution on $PG(k-1,q)$ is uniform. Equivalently, $p(0) = 0$ and $p(L) = 1/\qbin{k}{1}{q}$ for every $1$-dimensional subspace $L \le X$.
\end{thm}

To show this result, we need the following auxiliary lemma.
\begin{lem} \label{lem:simplex_optimality}
    Let $p$ be a probability distribution on $X$. For every $0 \le j \le k-1$, we have
    \[
    \sum_{\substack{V \le X \\ \dim(V) = j}} \sum_{\substack{S = (V_0, \dots, V_k) \in \mS \\ V_j = V}} p(V) \Pr[S] \ge \frac{\qbin{j}{1}{q}}{\qbin{k}{1}{q}}.
    \]
\end{lem}
\begin{proof}
From now on denote by 
\[
    A_j = \sum_{\substack{V \le X \\ \dim(V) = j}} \sum_{\substack{S = (V_0, \dots, V_k) \in \mS \\ V_j = V}} p(V) \Pr[S].
\]
We proceed by induction on $j$. For $j=0$, the only $0$-dimensional space is $V = 0$. Thus,
\[
\begin{split}
    A_0 &= \sum_{S = (0, V_1, \dots, V_k) \in \mS} p(0) \Pr[S] \\
    &= p(0) \sum_{S = (0, V_1, \dots, V_k) \in \mS} \Pr[S] \\
    &=p(0) \ge \frac{\qbin{0}{1}{q}}{\qbin{k}{1}{q}}.
\end{split}
\]
Assume now that $A_j \ge \qbin{j}{1}{q}/\qbin{k}{1}{q}$ for some $0 \le j \le k-2$. We fix a $j$-dimensional subspace $V$ as the $j$-th subspace of a chain $(V_0, V_1, \dots, V_k)$ and compute the expected probability mass of the next subspace in the chain, conditional on $V_j = V$. We have
\[
\begin{split}
    \E[p(V_{j+1}) \, | \, V_j = V] &= \sum_{\substack{U \ge V \\ \dim(U) = j+1}} p(U) \Pr[V_{j+1}= U \, | \, V_j = V] \\
    &= \sum_{\substack{U \ge V \\ \dim(U) = j+1}} p(U) \frac{p(U) - p(V)}{1 - p(V)} \\
    &= \frac{1}{1-p(V)} \sum_{\substack{U \ge V \\ \dim(U) = j+1}} (p(U)-p(V))^2 + \frac{p(V)}{1-p(V)} \sum_{\substack{U \ge V \\ \dim(U) = j+1}} (p(U)-p(V)) \\
    &= \frac{1}{1-p(V)} \sum_{\substack{U \ge V \\ \dim(U) = j+1}} (p(U)-p(V))^2 + p(V) \\
    &\ge \frac{1-p(V)}{\qbin{k-j}{1}{q}} + p(V),
\end{split}
\]
where we used Lemma~\ref{lem:dimension_increasing} in the second equality, the identity 
\[
\sum_{\substack{U \ge V\\ \dim(U) = j+1}}(p(U)-p(V)) = 1-p(V)
\]
in the second-to-last equality, and the Cauchy-Schwarz inequality
\[
\left( \sum_{\substack{U \ge V \\ \dim(U) = j+1}} (p(U)-p(V))^2 \right) \left( \sum_{\substack{U \ge V \\ \dim(U) = j+1}} 1^2 \right) \ge \left( \sum_{\substack{U \ge V \\ \dim(U) = j+1}} (p(U)-p(V)) \right)^2
\]
in the last step (recall that there are exactly $\qbin{k-j}{1}{q}$ subspaces of dimension $j+1$ containing a given $j$-dimensional subspace).

From now on, we focus on the inequality
\[
    \sum_{\substack{U \ge V \\ \dim(U) = j+1}} p(U) \frac{p(U) - p(V)}{1 - p(V)} \ge \frac{1-p(V)}{\qbin{k-j}{1}{q}} + p(V).
\]
Multiplying both sides by $\Pr[V_j = V]$ and then summing over all $j$-dimensional spaces yields
\begin{equation} \label{eq:auxiliary_lemma_simplex}
\begin{split}
    \sum_{\substack{V \le X \\ \dim(V) =j}} \Pr[V_j = V] &\sum_{\substack{U \ge V \\ \dim(U) = j+1}} p(U) \frac{p(U) - p(V)}{1 - p(V)} \\
    &\ge \sum_{\substack{V \le X \\ \dim(V) =j}} \Pr[V_j = V] \left( \frac{1-p(V)}{\qbin{k-j}{1}{q}} + p(V) \right).
\end{split}
\end{equation}
Let $L(j)$ and $R(j)$ denote the left hand side and the right hand side, respectively, of Equation~\eqref{eq:auxiliary_lemma_simplex}.
The law of total probability gives
\[
\begin{split}
    \Pr[V_{j+1}=U] &= \sum_{\substack{V \le U \\ \dim(V) = j}} \Pr[V_j = V] \Pr[V_{j+1}=U \, | \, V_j = V] \\
    &= \sum_{\substack{V \le U \\ \dim(V) = j}} \Pr[V_j = V] \frac{p(U)-p(V)}{1-p(V)},
\end{split}
\]
which can be applied to $L(j)$ as follows:
\[
\begin{split}
    L(j) &=  \sum_{\substack{U \le X \\ \dim(U) =j+1}} p(U)\sum_{\substack{V \le U \\ \dim(V) = j}} \Pr[V_j = V]  \frac{p(U) - p(V)}{1 - p(V)} \\
    &= \sum_{\substack{U \le X \\ \dim(U) =j+1}} p(U)\Pr[V_{j+1}=U] \\
    &= \sum_{\substack{U \le X\\ \dim(U) =j+1}} \sum_{\substack{S = (V_0, \dots, V_k) \in \mS \\ V_{j+1} = U}} p(U)\Pr[S] = A_{j+1}.
\end{split}
\]
On the other hand, 
\[
\begin{split}
    R(j) &=  \frac{1- \sum_{\substack{V \le X \\ \dim(V) =j}} p(V) \Pr[V_j = V]}{\qbin{k-j}{1}{q}} + \sum_{\substack{V \le X \\ \dim(V) =j}} p(V) \Pr[V_j = V]\\
    &=  \frac{1}{\qbin{k-j}{1}{q}} + \left( 1 -\frac{1}{\qbin{k-j}{1}{q}}  \right) \sum_{\substack{V \le X \\ \dim(V) =j}} p(V) \Pr[V_j = V] \\
    &= \frac{1}{\qbin{k-j}{1}{q}} + \left( 1 -\frac{1}{\qbin{k-j}{1}{q}}  \right) \sum_{\substack{V \le X \\ \dim(V) =j}} \sum_{\substack{S = (V_0, \dots, V_k) \in \mS \\ V_j = V}}  p(V) \Pr[S] \\
    &\ge \frac{1}{\qbin{k-j}{1}{q}} + \left( 1 -\frac{1}{\qbin{k-j}{1}{q}}  \right) \frac{\qbin{j}{1}{q}}{\qbin{k}{1}{q}} \\
    &= \frac{1-\qbin{j}{1}{q} / \qbin{k}{1}{q}}{\qbin{k-j}{1}{q}} + \frac{\qbin{j}{1}{q}}{\qbin{k}{1}{q}}  \\
    &= \frac{q^j}{\qbin{k}{1}{q}} + \frac{\qbin{j}{1}{q}}{\qbin{k}{1}{q}} = \frac{\qbin{j+1}{1}{q}}{\qbin{k}{1}{q}},
\end{split}
\]
where we used the induction hypothesis in the third-to-last step and the combinatorial identity
\[
\qbin{k}{1}{q} - \qbin{j}{1}{q} = q^j \qbin{k-j}{1}{q}
\]
in the second-to-last step.
Combining all of the above gives
\[
A_{j+1} \ge \frac{\qbin{j+1}{1}{q}}{\qbin{k}{1}{q}}.
\]
\end{proof}

\begin{proof}[Proof of Theorem \ref{thm:simplex_optimality}]
We can rewrite the formula from Theorem~\ref{thm:general_formula} conditioning on the $j$-th space in the chain $S = (V_0, \dots, V_k)$, namely
\begin{equation} \label{eq:gen_formula_expectation}
    \begin{split}
    \E[T] &= \sum_{j=0}^{k-1} \sum_{S = (V_0, \dots, V_k) \in \mS} \frac{\Pr[S]}{1-p(V_j)} \\
    &= \sum_{j=0}^{k-1} \sum_{\substack{V \le X \\ \dim(V) = j}} \sum_{\substack{S = (V_0, \dots, V_k) \in \mS \\ V_j = V}} \frac{\Pr[S]}{1-p(V)}.
\end{split}
\end{equation}
Since the function $\varphi(x) = 1/(1-x)$ is convex for $0 \le x <1$, Jensen's inequality gives
\begin{equation} \label{eq:optimization_jensen}
    \E[T] \ge \sum_{j=0}^{k-1}\frac{1}{1-\sum_{\substack{V \le X \\ \dim(V) = j}} \sum_{\substack{S = (V_0, \dots, V_k) \in \mS \\ V_j = V}} p(V) \Pr[S]  }.
\end{equation}
Applying Lemma~\ref{lem:simplex_optimality} to Equation~\eqref{eq:optimization_jensen} yields
\[
    \E[T]  \ge \sum_{j=0}^{k-1} \frac{1}{1-\frac{\qbin{j}{1}{q}}{\qbin{k}{1}{q}}} = \sum_{j=0}^{k-1} \frac{q^k-1}{q^k-q^j}.
\]

It remains to characterize equality. 
From now on, for $0\leq j\leq k-1$, let 
\[
B_j = \sum_{\substack{V\leq X\\\dim(V)=j}} \sum_{\substack{S\in\mathcal{S}\\V_j=V}} \frac{\Pr[S]}{1-p(V)}. 
\] 
The preceding argument shows, term by term, that 
\begin{equation} \label{eq:inequalities_chain}
    B_j \geq \frac{1}{1-A_j} \geq \frac{1}{ 1- \dfrac{\qbin{j}{1}{q}} {\qbin{k}{1}{q}} }. 
\end{equation}
Suppose that equality holds in~\eqref{eq:simplex_lower_bound}. 
Therefore, the inequalities in~\eqref{eq:inequalities_chain} must also be equalities. For $j=0$, we have $V_0=\{0\}$ with probability one and hence
\[ B_0=\frac{1}{1-p(0)}. \] 
The corresponding term in the lower bound is $1$, which implies $p(0)=0$.

We next consider $j=1$. Since $p(0)=0$, the first draw is nonzero with probability one. Thus, for every $L\in\mathcal{P}_k$, we have $\Pr[V_1=L]=p(L).$ 
It follows that 
\[ A_1 
= \sum_{L\in\mathcal{P}_k}p(L)^2. \] 
The fact that the equalities must hold in~\eqref{eq:inequalities_chain} implies that $A_1 = 1/\qbin{k}{1}{q}.$ 
Therefore, 
\[ \sum_{L\in\mathcal{P}_k}p(L)^2 = \frac{1}{\qbin{k}{1}{q}}. \] 
On the other hand, 
the Cauchy-Schwarz inequality gives
\[ \sum_{L\in\mathcal{P}_k}p(L)^2 \geq \frac{1}{|\mathcal{P}_k|} \left(\sum_{L\in\mathcal{P}_k}p(L)\right)^2 = \frac{1}{\qbin{k}{1}{q}}, \] 
and equality holds if and only if all the values $p(L)$ are equal. Consequently, $p(L)=1/\qbin{k}{1}{q}$ for every $L\in\mathcal{P}_k$.

Conversely, suppose that $p(0)=0$ and that $p(L)=1/\qbin{k}{1}{q}$ 
for every one-dimensional subspace $L\leq X$. If $V\leq X$ has dimension $j$, then $V$ contains exactly $\qbin{j}{1}{q}$ one-dimensional subspaces. Therefore, 
\[ p(V) = \frac{\qbin{j}{1}{q}}{\qbin{k}{1}{q}}. \] 
In particular, $p(V)$ depends only on $\dim(V)$. Substituting back into~\eqref{eq:gen_formula_expectation}, we obtain
\begin{align*} 
\mathbb{E}[T] &= \sum_{j=0}^{k-1} \frac{1}{ 1- \dfrac{\qbin{j}{1}{q}} {\qbin{k}{1}{q}} } \sum_{S\in\mathcal{S}}\Pr[S]\\ &= \sum_{j=0}^{k-1} \frac{q^k-1}{q^k-q^j}. \end{align*} 
Thus the lower bound is attained, completing the proof. 
\end{proof}

Recall that a generator matrix $G$ of the $k$-dimensional $q$-ary simplex code is obtained by organizing all nonzero elements of $\F_q^k$, up to nonzero scalar multiples, as columns of $G$. Equivalently, each projective point of $PG(k-1,q)$ appears exactly once. This fact is crucial to prove their optimality.

\begin{cor}\label{cor:simplex-unique-optimal-code} 
Let $n=(q^k-1)/(q-1)$, and let $G\in\mathbb{F}_q^{k\times n}$ have rank $k$. Then 
\[
\mathbb{E}[G] \geq \sum_{j=0}^{k-1} \frac{q^k-1}{q^k-q^j}. 
\] 
In particular, equality holds if and only if $G$ is a generator matrix of the $q$-ary simplex code.
Equivalently,
\[
\mathbb{E}_{\text{opt}}\left[\frac{q^k-1}{q-1},k\right]_q = \sum_{j=0}^{k-1} \frac{q^k-1}{q^k-q^j}.
\]
Consequently, every optimal $\left[\frac{q^k-1}{q-1},k\right]_q$ code is monomially equivalent to the $q$-ary simplex code. 
\end{cor}

\begin{rem}
    A proof showing that the simplex code is optimal in terms of the expectation was independently found by Zhihan Jin and Barnabas Janzer (private communication). In this paper, the optimality of the simplex code is obtained as a corollary of a more general statement (Theorem~\ref{thm:simplex_optimality}).
\end{rem}

\begin{proof}[Proof of Corollary~\ref{cor:simplex-unique-optimal-code}]
Recall that the probability distribution $p_G$ induced by the drawing procedure is given by 
\[ p_G(x) = \frac{|\{i\in[n]:g_i=x\}|}{n}. \] 
By Theorem~\ref{thm:simplex_optimality}, equality holds if and only if $p_G(0)=0$ and
\[
p_G(L)=\frac{1}{\qbin{k}{1}{q}} =\frac{1}{n}
\]
for every $1$-dimensional subspace $L$.
The latter conditions imply that $z_G=0$ and $m_G(L)=1$ for every projective point $L$. Therefore, $G$ contains exactly one nonzero representative of every one-dimensional subspace of $\mathbb{F}_q^k$, which is precisely the defining property of a $q$-ary simplex matrix. 

To conclude, any two simplex matrices differ by a permutation of columns and by multiplication of individual columns by nonzero scalars.
However, these operations produce monomially equivalent codes. Therefore, all optimal codes with these parameters are monomially equivalent to the $q$-ary simplex code. 
\end{proof}

\section{Optimality of First-Order $q$-ary Reed--Muller Codes}
\label{section:ReedMuller}

In this section, we prove that the first-order $q$-ary Reed--Muller code is optimal with respect to the coverage depth. The proof relies on the optimality of the $q$-ary simplex code established in Section~\ref{section:SimplexCode}. The main idea is that, after the first nonzero column is sampled, the remaining coverage process can be viewed in a quotient space of dimension one less.

\begin{thm}
\label{thm:reed_muller_optimality}
Let $m\geq 1$, let $n=q^m$, and let $R\in\F_q^{(m+1)\times n}$ be a generator matrix of the first-order $q$-ary Reed--Muller code. Then, for every matrix $A\in\F_q^{(m+1)\times n}$ of rank $m+1$,
\[
\E[A]\geq \E[R].
\]
In particular,
\[
\E_{\mathrm{opt}}[q^m,m+1]_q
=
m+1+\sum_{j=1}^{m}\frac{1}{q^j-1}.
\]
\end{thm}

\begin{proof}
Recall that $z_A=\left|\{i\in[n]:a_i=0\}\right|$.
Let $\tau$ be the first time that a nonzero column is sampled. Since a sampled column is nonzero with probability $(n-z_A)/n$, we have
\[
\E[\tau]=\frac{n}{n-z_A}.
\]

Suppose that the first nonzero column sampled is $a_j$. Consider the quotient space
\[
X_j=\F_q^{m+1}/\langle a_j\rangle,
\]
which has dimension $m$, and let
\[
\pi_j:\F_q^{m+1}\longrightarrow X_j
\]
be the natural projection. Once $a_j$ has been sampled, spanning $\F_q^{m+1}$ is equivalent to having the projections of the subsequently sampled columns span $X_j$.

Let
\[
d_j=\left|\{i\in[n]:a_i\in\langle a_j\rangle\}\right|.
\]
A projected column is zero precisely when the sampled column belongs to $\langle a_j\rangle$. Hence, a nonzero projected column is obtained with probability
\[
\frac{n-d_j}{n}.
\]
Moreover, since all zero columns belong to $\langle a_j\rangle$ and $a_j\neq0$, we have
\[
d_j\geq z_A+1.
\]

Conditioned on obtaining a nonzero projection, the projected columns induce a probability distribution on $X_j\setminus\{0\}$. Since $A$ has rank $m+1$, the support of this distribution spans $X_j$. Therefore, by Theorem~\ref{thm:simplex_optimality}, the expected number of nonzero projected columns needed to span $X_j$ is at least the coverage depth of the $q$-ary simplex code of dimension $m$. Denote a generator matrix of this simplex code by $S_m$.

Since a nonzero projected column is obtained with probability $(n-d_j)/n$, it follows that, conditioned on $a_j$ being the first nonzero sampled column, the expected number of additional draws is at least
\[
\frac{n}{n-d_j}\E[S_m].
\]
Combining this with $d_j\geq z_A+1$ yields
\[
\E[A]
\geq
\frac{n}{n-z_A}
+
\frac{n}{n-z_A-1}\E[S_m].
\]
Consequently,
\[
\E[A]
\geq
1+\frac{n}{n-1}\E[S_m].
\]

It remains to show that equality is attained by the first-order Reed--Muller code. A natural generator matrix for this code has columns
\[
r_v=(1,v),\qquad v\in\F_q^m.
\]
In particular, it has no zero columns. Fix a column $r_{v_0}=(1,v_0)$. Since all columns have first coordinate equal to $1$, no other column lies in $\langle r_{v_0}\rangle$, and hence $d_j=1$.

Furthermore, modulo $\langle r_{v_0}\rangle$, the column $r_v=(1,v)$ can be represented by
\[
r_v-r_{v_0}=(0,v-v_0).
\]
As $v$ ranges over $\F_q^m\setminus\{v_0\}$, the vectors $v-v_0$ range over all nonzero elements of $\F_q^m$ exactly once. Thus, conditioned on obtaining a nonzero projection, the induced distribution on the one-dimensional subspaces of $X_j$ is uniform. By the equality condition in Theorem~\ref{thm:simplex_optimality}, the corresponding projected coverage process attains $\E[S_m]$.

Therefore,
\[
\E[R]
=
1+\frac{n}{n-1}\E[S_m].
\]
Since $n=q^m$ and
\[
\E[S_m]
=
\sum_{\ell=0}^{m-1}
\frac{q^m-1}{q^m-q^\ell},
\]
we obtain
\[
\begin{split}
\E[R]
&=
1+\sum_{\ell=0}^{m-1}
\frac{q^m}{q^m-q^\ell}\\
&=
1+\sum_{j=1}^{m}\frac{q^j}{q^j-1} =
m+1+\sum_{j=1}^{m}\frac{1}{q^j-1}.
\end{split}
\]
Hence $\E[A]\geq\E[R]$, completing the proof.
\end{proof}

\section{Optimality of the $q$-ary Hamming Code}\label{section:HammingCode}

In this section, we prove that the \(q\)-ary Hamming code is optimal, i.e., it minimizes the expected coverage time among all linear codes with the same length and dimension. We begin with a formula that expresses the expected coverage time in terms of spanning subsets of the columns of a generator matrix. By passing to the dual code, these quantities become counts of linearly independent column subsets. We then represent these counts using the basis generating polynomial of a projective-geometry matroid and show, by averaging over the action of \(\mathrm{GL}(r,q)\), that this polynomial is maximized by the simplex matrix.

\begin{thm}\label{th:Hamm}
Let \(r\geq 2\), and let \(H\in\mathbb{F}_q^{k\times n}\) be a generator matrix of the \(q\)-ary Hamming code $\mathcal{C} \subseteq \F_q^n$, where \(n=(q^r-1)/(q-1)\) and \(k=n-r\). For every matrix \(A\in\mathbb{F}_q^{k\times n}\) of rank \(k\), we have
\(
\mathbb{E}[A]\geq \mathbb{E}[H].
\) 
In particular, 
\[
\E_{\text{opt}}\left[\frac{q^r-1}{q-1}, \frac{q^r-1}{q-1}-r \right]_q = \E [H].
\]
\end{thm}

Recall, for \(k\leq s\leq n\), the quantities
\[
\alpha(G,s)=\left|\left\{S\subseteq[n]: |S|=s,\ \operatorname{rank}(G_S)=k\right\}\right|,
\]
which counts the number of $s$-subsets of columns of $G$ that span the space $\F_q^k$. 

Gruica et al.~\cite{gruica2025combinatorial} showed that it is possible to express the expectation in terms of the latter quantities as follows:
\[
\mathbb{E}[G]=nH_n-\sum_{s=k}^{n-1}\frac{\alpha(G,s)}{\binom{n-1}{s}}.
\]
Since all coefficients in the sum are positive, minimizing \(\mathbb{E}[G]\) amounts to maximizing the quantities \(\alpha(G,s)\).
The spanning condition in the definition of \(\alpha(G,s)\) has a convenient dual interpretation. To state it, we introduce the following notation.

\begin{defn}
Let \(G\in\mathbb{F}_q^{k\times n}\). For \(0\leq t\leq n\), define
\[
I(G,t)=\left|\left\{T\subseteq[n]: |T|=t,\ \operatorname{rank}(G_T)=t\right\}\right|.
\]
Thus, \(I(G,t)\) is the number of linearly independent subsets of \(t\) columns of \(G\).
\end{defn}
We will use the following standard relation between spanning sets of a code and independent sets of its dual. We include the short proof for completeness.
\begin{lem}
For every \(k\leq s\leq n\),
\(
\alpha(G,s)=I(G^\perp,n-s).
\)
\end{lem}

\begin{proof}
Fix \(S\subseteq[n]\), and let \(T=[n]\setminus S\). The matrix \(G_S\) has rank \(k\) if and only if there is no nonzero codeword of \(\mathcal{C}_G\) whose support is contained in \(T\). On the other hand, the columns of \(G^\perp\) indexed by \(T\) are linearly dependent if and only if there exists a nonzero vector supported on \(T\) that belongs to
\[
\ker(G^\perp)=\mathcal{C}_G.
\]
Therefore, \(G_S\) has rank \(k\) if and only if the columns of \((G^\perp)_T\) are linearly independent. The correspondence \(S\mapsto[n]\setminus S\) is a bijection between subsets of size \(s\) and subsets of size \(n-s\), which proves the claim.
\end{proof}

Substituting \(t=n-s\) into the formula for the expected coverage time gives
\[
\mathbb{E}[G]=nH_n-\sum_{t=1}^{n-k}\frac{I(G^\perp,t)}{\binom{n-1}{t-1}}.
\]
Hence, minimizing the expected coverage time reduces to maximizing \(I(G^\perp,t)\) for each \(1\leq t\leq r\). To analyze these quantities, we index the nonzero column directions by the one-dimensional subspaces of \(\mathbb{F}_q^r\), record their multiplicities in a vector, and describe their linear-dependence structure using a matroid.

The linear-dependence structure of the nonzero column directions is described by a matroid on the projective space \(\mathcal{P}_r\). We recall the matroid terminology needed below.

\begin{defn}
A matroid is a pair \(M=(E,\mathcal{I})\), where \(E\) is a finite set and \(\mathcal{I}\) is a nonempty collection of subsets of \(E\) satisfying the following properties:
\begin{enumerate}
    \item If \(T\in\mathcal{I}\) and \(S\subseteq T\), then \(S\in\mathcal{I}\).
    \item If \(S,T\in\mathcal{I}\) and \(|S|<|T|\), then there exists \(i\in T\setminus S\) such that \(S\cup\{i\}\in\mathcal{I}\).
\end{enumerate}
The sets in \(\mathcal{I}\) are called independent sets. A basis is a maximal independent set, and the rank of \(M\) is the cardinality of any basis.
\end{defn}

For \(1\leq t\leq r\), let \(M_{r,t}\) be the matroid on \(\mathcal{P}_r\) whose independent sets are the linearly independent subsets of size at most \(t\). Its bases are precisely the linearly independent subsets of \(\mathcal{P}_r\) of size \(t\).

\begin{defn}
Let \(M=(E,\mathcal{I})\) be a matroid, and let \(\mathcal{B}_M\) denote its collection of bases. The basis generating polynomial of \(M\) is
\[
g_M(\mathbf{z})=\sum_{B\in\mathcal{B}_M}\prod_{e\in B}z_e.
\]
\end{defn}

For the matroid \(M_{r,t}\), we write
\[
g_{r,t}(\mathbf{z})=\sum_{\substack{P\subseteq\mathcal{P}_r\\ |P|=t,\ P\text{ is linearly independent}}}\prod_{p\in P}z_p.
\]
The polynomial \(g_{r,t}\) is homogeneous of degree \(t\). Moreover, by Theorem~4.2 of Anari et al.~\cite{anari2018log}, the basis generating polynomial of every matroid is completely log-concave over the positive orthant. In particular, \(\log g_{r,t}\) is concave wherever \(g_{r,t}\) is positive.

We now encode the column multiplicities of a matrix as a vector indexed by \(\mathcal{P}_r\). 
Recall that, for a matrix \(G\in\mathbb{F}_q^{r\times n}\), and for each \(p\in\mathcal{P}_r\), \(m_G(p)\) is the number of nonzero columns of \(G\) whose span is \(p\).
The basis generating polynomial evaluated at \(m_G\) counts the independent subsets of columns of \(G\).

\begin{lem}
For every \(1\leq t\leq r\),
\[
I(G,t)=g_{r,t}(m_G).
\]
\end{lem}

\begin{proof}
Every linearly independent subset of \(t\) columns of \(G\) consists of nonzero columns with distinct and linearly independent spans. It therefore determines a linearly independent set \(P\subseteq\mathcal{P}_r\) of size \(t\).

Conversely, fix a linearly independent set \(P\subseteq\mathcal{P}_r\) of size \(t\). For each \(p\in P\), one may choose any of the \(m_G(p)\) columns of \(G\) whose span is \(p\). Hence, the number of independent column subsets associated with \(P\) is \(\prod_{p\in P}m_G(p)\). Summing over all linearly independent sets \(P\) of size \(t\) gives
\[
I(G,t)=\sum_{\substack{P\subseteq\mathcal{P}_r\\ |P|=t,\ P\text{ is linearly independent}}}\prod_{p\in P}m_G(p)=g_{r,t}(m_G).
\]
\end{proof}

We have now reduced the theorem to an inequality for \(g_{r,t}\). For a parity-check matrix \(B\), the value \(g_{r,t}(m_B)\) counts the independent subsets of \(t\) columns. For a simplex matrix \(S\), every projective point appears exactly once, so \(m_S=\mathbf{1}\). It remains to show that \(g_{r,t}\) is maximized at this uniform multiplicity vector.

\begin{proof}[Proof Theorem~\ref{th:Hamm}]
Let \(A\in\mathbb{F}_q^{k\times n}\) have rank \(k\), and let \(B\in\mathbb{F}_q^{r\times n}\) be a parity-check matrix of \(\mathcal{C}_A\). Let \(S\) be a \(q\)-ary simplex matrix of dimension \(r\). The columns of \(S\) contain exactly one nonzero representative of each element of \(\mathcal{P}_r\). Therefore, \(m_S=\mathbf{1}\), and it is sufficient to prove that
\[
g_{r,t}(m_B)\leq g_{r,t}(\mathbf{1})
\]
for every \(1\leq t\leq r\).

The group \(\mathrm{GL}(r,q)\) acts transitively on \(\mathcal{P}_r\). For \(\sigma\in\mathrm{GL}(r,q)\) and \(x\in\mathbb{R}^{\mathcal{P}_r}\), define
\[
(\sigma x)(p)=x(\sigma^{-1}p).
\]
Since every \(\sigma\in\mathrm{GL}(r,q)\) preserves linear independence, it permutes the bases of \(M_{r,t}\). Consequently,
\[
g_{r,t}(\sigma x)=g_{r,t}(x).
\]

Define the orbit average
\[
\overline{m}_B=\frac{1}{|\mathrm{GL}(r,q)|}\sum_{\sigma\in\mathrm{GL}(r,q)}\sigma m_B.
\]
If \(g_{r,t}(m_B)=0\), then the desired inequality is immediate. Assume therefore that \(g_{r,t}(m_B)>0\).

For every \(\varepsilon>0\), the vector \(m_B+\varepsilon\mathbf{1}\) lies in the positive orthant. Since \(\mathbf{1}\) is fixed by the action of \(\mathrm{GL}(r,q)\), the invariance of \(g_{r,t}\) and Jensen's inequality imply
\[
\log g_{r,t}(m_B+\varepsilon\mathbf{1})\leq \log g_{r,t}(\overline{m}_B+\varepsilon\mathbf{1}).
\]
The logarithm is monotonically increasing, and therefore
\[
g_{r,t}(m_B+\varepsilon\mathbf{1})\leq g_{r,t}(\overline{m}_B+\varepsilon\mathbf{1}).
\]
Letting \(\varepsilon\rightarrow0\) and using the continuity of \(g_{r,t}\), we obtain
\[
g_{r,t}(m_B)\leq g_{r,t}(\overline{m}_B).
\]

The transitivity of the action of \(\mathrm{GL}(r,q)\) on \(\mathcal{P}_r\) implies that all coordinates of \(\overline{m}_B\) are equal. Hence,
\[
\overline{m}_B=c\mathbf{1}
\]
for some \(c\geq0\).

The action of \(\mathrm{GL}(r,q)\) only permutes the coordinates of \(m_B\), so
\[
\sum_{p\in\mathcal{P}_r}\overline{m}_B(p)=\sum_{p\in\mathcal{P}_r}m_B(p)=\left|\left\{i\in[n]: b_i\neq0\right\}\right|.
\]
Since \(|\mathcal{P}_r|=n\), it follows that
\[
c=\frac{\left|\left\{i\in[n]: b_i\neq0\right\}\right|}{n}\leq1.
\]
In particular, \(c=1\) if \(B\) has no zero columns, while \(c<1\) otherwise.

Since \(g_{r,t}\) is homogeneous of degree \(t\),
\[
g_{r,t}(\overline{m}_B)=g_{r,t}(c\mathbf{1})=c^t g_{r,t}(\mathbf{1})\leq g_{r,t}(\mathbf{1})=g_{r,t}(m_S).
\]
Combining the two inequalities gives
\[
I(B,t)=g_{r,t}(m_B)\leq g_{r,t}(m_S)=I(S,t)
\]
for every \(1\leq t\leq r\).

Finally, the dual expression for the expected coverage time yields
\[
\mathbb{E}[A]=nH_n-\sum_{t=1}^{r}\frac{I(B,t)}{\binom{n-1}{t-1}}
\geq nH_n-\sum_{t=1}^{r}\frac{I(S,t)}{\binom{n-1}{t-1}}
=\mathbb{E}[H],
\]
which is the desired result.
\end{proof}

\section{Uniformly Random Codes}\label{section:Random Codes}

We now study the expected coverage depth of a random linear code. Random matrices were previously used by Wang and Yaakobi~\cite{wang2026randomaccess} as a probabilistic construction to derive upper bounds on the minimum achievable coverage depth. Here, rather than using randomness only as an existence argument, our objective is to characterize the random-code ensemble itself. We sample uniformly from the set of all full-rank $k\times n$ matrices and determine its average coverage depth exactly.

Let
\[
\mathcal{M}_{k,n}(q)
=
\left\{
A\in\mathbb{F}_q^{k\times n}:
\rank(A)=k
\right\},
\]
and let \(A\) be sampled uniformly from \(\mathcal{M}_{k,n}(q)\).

Conditioned on \(A\), let \(T_A\) be its coverage time as defined in the preliminaries. Unless stated otherwise, all expectations in this section are taken jointly over the random choice of \(A\) and the column-sampling process.
For brevity, we write \(T=T_A\).
We first derive an exact expression for this expectation.

\begin{lem}
\label{lem:random-exact}
Let \(A\) be sampled uniformly from \(\mathcal{M}_{k,n}(q)\). Then
\begin{equation}
\label{eq:random-exact}
\begin{aligned}
\mathbb{E}[T]
&=
n\left(H_n-H_{n-k}\right) \\
&\quad+
n\sum_{t=k}^{n-1}
\frac{1}{n-t}
\left(
1-
q^{k(n-t)}
\prod_{i=0}^{k-1}
\frac{q^t-q^i}{q^n-q^i}
\right).
\end{aligned}
\end{equation}

Equivalently,
\begin{equation}
\label{eq:random-exact-product}
\mathbb{E}[T]
=
n\left(H_n-H_{n-k}\right)
+
n\sum_{t=k}^{n-1}
\frac{1}{n-t}
\left(
1-
\prod_{i=0}^{k-1}
\left(
1-
\frac{q^{i-t}-q^{i-n}}
{1-q^{i-n}}
\right)
\right).
\end{equation}
\end{lem}

\begin{proof}
Let \(D\) denote the number of distinct column indices sampled by the
time the sampled columns first attain rank \(k\). Conditioned on \(D=d\), the number of samples required is the
waiting time to observe \(d\) distinct elements in the classical
coupon-collector process. Hence, by the partial coupon-collector
formula \cite[p.~216]{flajolet1992birthday},
\[
\mathbb{E}[T\mid D=d]
=
n\left(H_n-H_{n-d}\right).
\]
By the law of total expectation,
\[
\mathbb{E}[T]
=
\sum_{d=0}^{n}
n\left(H_n-H_{n-d}\right)\Pr(D=d).
\]
Since
\[
H_n-H_{n-d}
=
\sum_{j=n-d+1}^{n}\frac{1}{j},
\]
we obtain
\[
\mathbb{E}[T]
=
\sum_{d=0}^{n}
n
\left(
\sum_{j=n-d+1}^{n}\frac{1}{j}
\right)
\Pr(D=d).
\]

Interchanging the order of summation gives
\[
\mathbb{E}[T]
=
n\sum_{j=1}^{n}
\frac{1}{j}
\Pr(D\geq n-j+1).
\]
Setting \(t=n-j\), we obtain
\[
\mathbb{E}[T]
=
n\sum_{t=0}^{n-1}
\frac{1}{n-t}\Pr(D>t).
\]

Since at least \(k\) distinct columns are required to attain rank \(k\),
\[
\Pr(D>t)=1,
\qquad
0\leq t\leq k-1.
\]
Consequently,
\begin{equation}
\label{eq:random-tail}
\mathbb{E}[T]
=
n\left(H_n-H_{n-k}\right)
+
n\sum_{t=k}^{n-1}
\frac{1}{n-t}\Pr(D>t).
\end{equation}

It remains to compute \(\Pr(D>t)\). The first \(t\) distinct sampled
indices form a uniformly random ordered \(t\)-tuple of distinct
indices. Moreover, the distribution of \(A\) is invariant under column
permutations. Therefore, if \(A_{[t]}\) denotes the submatrix formed by
the first \(t\) columns of \(A\), then
\[
\Pr(D>t)
=
\Pr\left(\rank(A_{[t]})<k\right).
\]
Hence,
\[
\Pr(D>t)
=
1-\Pr\left(\rank(A_{[t]})=k\right).
\]

The number of full-rank matrices in
\(\mathbb{F}_q^{k\times n}\) is
\[
\left|\mathcal{M}_{k,n}(q)\right|
=
\prod_{i=0}^{k-1}(q^n-q^i).
\]
Indeed, after choosing \(i\) linearly independent rows, their span has
cardinality \(q^i\), and the next row can be chosen in
\(q^n-q^i\) ways.

To count the full-rank matrices whose first \(t\) columns already have
rank \(k\), write
\[
A=[\,B\mid C\,],
\]
where
\[
B\in\mathbb{F}_q^{k\times t},
\qquad
C\in\mathbb{F}_q^{k\times(n-t)}.
\]
The matrix \(B\) must have rank \(k\), whereas \(C\) may be arbitrary.
Hence, the number of such matrices is
\[
q^{k(n-t)}
\prod_{i=0}^{k-1}(q^t-q^i).
\]

It follows that
\[
\Pr\left(\rank(A_{[t]})=k\right)
=
q^{k(n-t)}
\prod_{i=0}^{k-1}
\frac{q^t-q^i}{q^n-q^i}.
\]
Substituting this into \eqref{eq:random-tail} proves
\eqref{eq:random-exact}.

Finally,
\[
q^{n-t}
\frac{q^t-q^i}{q^n-q^i}
=
\frac{1-q^{i-t}}{1-q^{i-n}}
=
1-
\frac{q^{i-t}-q^{i-n}}
{1-q^{i-n}}.
\]
Distributing \(q^{k(n-t)}\) over the \(k\) factors in the product gives
\eqref{eq:random-exact-product}.
\end{proof}

The first term in \eqref{eq:random-exact} is the coverage depth of an
MDS code. Indeed, every \(k\) distinct columns of an MDS generator
matrix are linearly independent, and therefore
\[
\mathbb{E}[T_{\mathrm{MDS}}]
=
n\left(H_n-H_{n-k}\right).
\]

We define the additive gap from the MDS benchmark.
\begin{defn}
\[
\Delta_{n,k,q}
\triangleq
\mathbb{E}[T]
-
n\left(H_n-H_{n-k}\right).
\]
\end{defn}
Next, we derive an upper bound on this gap. The analysis uses the
following auxiliary sequence.

\begin{defn}
For \(m\geq1\), define
\[
S_m(q)
=
\sum_{j=1}^{m}
\frac{m}{j}q^{j-m}.
\]    
\end{defn}

\begin{cla}
\label{claim:Sm}
\begin{equation}
\label{eq:Sm-recurrence}
S_{m+1}(q)
=
1+\frac{m+1}{qm}S_m(q).
\end{equation}
Moreover, the sequence \(\{S_m(q)\}_{m\geq1}\) is bounded and
\[
\lim_{m\to\infty}S_m(q)
=
\frac{q}{q-1}.
\]
\end{cla}

\begin{proof}
We have
\[
\begin{split}
    S_{m+1}(q) &= \sum_{j=1}^{m+1} \frac{m+1}{j}q^{j-m-1}
    = 1+ \sum_{j=1}^{m} \frac{m+1}{j}q^{j-m-1} \\
    &= 1+ \frac{m+1}{qm} \sum_{j=1}^{m} \frac{m}{j}q^{j-m} = 1+\frac{m+1}{qm}S_m(q),
\end{split}
\]
which proves \eqref{eq:Sm-recurrence}.

We next show that \(S_m(q)\) is bounded. Setting \(r=m-j\), we obtain
\[
S_m(q)
=
\sum_{r=0}^{m-1}
\frac{m}{m-r}q^{-r}.
\]
Note that, for \(0\leq r\leq m-1\),
\[
\frac{m}{m-r}\leq 1+r.
\]
Equivalently, 
\[
m
\leq
(1+r)(m-r)
=
m+r(m-r-1).
\]

It follows that
\[
S_m(q)
\leq
\sum_{r=0}^{m-1}(1+r)q^{-r}
\leq
\sum_{r=0}^{\infty}(1+r)q^{-r}
<\infty.
\]
Thus, the sequence is bounded.

Let
\[
\overline{S}
=
\limsup_{m\to\infty}S_m(q),
\qquad
\underline{S}
=
\liminf_{m\to\infty}S_m(q).
\]
Since shifting a sequence does not change its limit superior and
\(
\frac{m+1}{qm}\longrightarrow\frac{1}{q},
\)
Equation~\eqref{eq:Sm-recurrence} gives
\(
\overline{S}
=
1+\frac{1}{q}\overline{S}.
\)
Therefore,
\(
\overline{S}
=
\frac{q}{q-1}.
\) Similarly,
\(
\underline{S}
=
1+\frac{1}{q}\underline{S},
\)
and hence
\(
\underline{S}
=
\frac{q}{q-1}.
\)
Thus,
\(
\underline{S}
=
\overline{S}
=
\frac{q}{q-1},
\)
which proves that
\(
S_m(q)\longrightarrow\frac{q}{q-1}.
\)
\end{proof}

The following lemma relates \(S_m(q)\) to the gap from the MDS
benchmark.

\begin{lem}
\label{lem:random-gap-bound}
Let \(m=n-k\). Then
\begin{equation}
\label{eq:random-gap-bound}
\Delta_{n,k,q}
\leq
\frac{1}{1-q^{-m}}
\cdot
\frac{S_m(q)}{q-1}
\cdot
\frac{n}{m}.
\end{equation}
\end{lem}

\begin{proof}
For \(k\leq t\leq n-1\) and \(0\leq i<k\), define
\[
a_i
=
\frac{q^{i-t}-q^{i-n}}{1-q^{i-n}}.
\]
Since \(k\leq t<n\) and \(i<k\), we have $0<a_i<1.$
Using Weierstrass Product Inequality,
equation~\eqref{eq:random-exact-product} gives
\[
\Delta_{n,k,q}
\leq
n\sum_{t=k}^{n-1}
\frac{1}{n-t}
\sum_{i=0}^{k-1}
\frac{q^{i-t}-q^{i-n}}
{1-q^{i-n}}
\le
n\sum_{t=k}^{n-1}
\frac{1}{n-t}
\sum_{i=0}^{k-1}
\frac{q^{i-t}}
{1-q^{i-n}}.
\]

Since \(i\leq k-1\),
\[
\frac{1}{1-q^{i-n}}
\leq
\frac{1}{1-q^{k-n}}.
\]
Therefore,
\[
\Delta_{n,k,q}
\leq
\frac{n}{1-q^{k-n}}
\sum_{t=k}^{n-1}
\frac{1}{n-t}
\sum_{i=0}^{k-1}q^{i-t}.
\]

Furthermore,
\[
\sum_{i=0}^{k-1}q^{i-t}
=
q^{-t}\frac{q^k-1}{q-1}
\leq
\frac{q^{k-t}}{q-1}.
\]
Thus,
\[
\Delta_{n,k,q}
\leq
\frac{n}{1-q^{k-n}}
\cdot
\frac{1}{q-1}
\sum_{t=k}^{n-1}
\frac{q^{k-t}}{n-t}.
\]

Set \(m=n-k\) and change variables according to \(j=n-t\). Then
\[
\sum_{t=k}^{n-1}
\frac{q^{k-t}}{n-t}
=
\sum_{j=1}^{m}
\frac{q^{j-m}}{j}
=
\frac{S_m(q)}{m}.
\]
Since \(q^{k-n}=q^{-m}\), we obtain
\[
\Delta_{n,k,q}
\leq
\frac{1}{1-q^{-m}}
\cdot
\frac{S_m(q)}{q-1}
\cdot
\frac{n}{m}.
\]
\end{proof}

We can now derive the asymptotic behavior in the two main regimes.

\begin{thm}[Constant-rate regime]
\label{thm:random-constant-rate}
Let \(q\) and \(R\in(0,1)\) be fixed, and suppose that
\[
\frac{k}{n}=R.
\]
Then, as \(n\to\infty\),
\[
\mathbb{E}[T]
\leq
n\left(H_n-H_{n-k}\right)
+
\frac{q}{(q-1)^2}\cdot\frac{1}{1-R}
+
o(1).
\]
Equivalently, using \(n=k/R\) and the asymptotic expansion of the
harmonic numbers,
\[
\mathbb{E}[T]
\leq
n\log\left(\frac{1}{1-R}\right)
+
\frac{1}{1-R}
\left(
\frac{q}{(q-1)^2}-\frac{R}{2}
\right)
+
o(1).
\]

In particular, the additive gap between the expected coverage depth of
a random code and the MDS benchmark is bounded independently of \(n\).
\end{thm}

For comparison, Wang and Yaakobi~\cite[Corollary~6]{wang2026randomaccess} showed that the optimal coverage depth in the same constant-rate regime is within a constant additive gap of the MDS benchmark, with the explicit upper bound
\[
\Delta_{n,k,q}
\leq
2\cdot\frac{q}{(q-1)^2}\cdot\frac{1}{1-R}+o(1).
\]
Our result concerns the ensemble-average performance rather than the existence of a favorable code: it shows that the average over all full-rank matrices satisfies
\[
\Delta_{n,k,q}\leq
\frac{q}{(q-1)^2}\cdot\frac{1}{1-R}+o(1).
\]
Thus, the entire full-rank ensemble achieves a bounded average additive gap, with an explicit constant that is a factor of two smaller than the one stated in~\cite[Corollary~6]{wang2026randomaccess}.

\begin{proof}
Let \(m=n-k\). Since \(k/n=R\), we have
\(
m=n(1-R).
\) Hence, \(m\to\infty\) as \(n\to\infty\), and
\(
\frac{n}{m}=\frac{1}{1-R}.
\) Moreover, since \(m\to\infty\), Lemma~\ref{claim:Sm} gives
\(
S_m(q)=\frac{q}{q-1}+o(1),
\)
while
\(
\frac{1}{1-q^{-m}}=1+o(1).
\)
Applying Lemma~\ref{lem:random-gap-bound}, we obtain
\[
\Delta_{n,k,q}
\leq
\frac{q}{(q-1)^2}\cdot\frac{1}{1-R}
+
o(1),
\]
which proves the result.
\end{proof}

\begin{thm}[Constant-redundancy regime]
\label{thm:random-constant-redundancy}
Let \(q\) and \(c\geq1\) be fixed, and suppose that
\[
n-k=c.
\]
Then, as \(n\to\infty\),
\[
\mathbb{E}[T]
\leq
n\left(H_n-H_c\right)
+
\frac{1}{1-q^{-c}}
\cdot
\frac{S_c(q)}{q-1}
\cdot
\frac{n}{c}
=
n\left(H_n-H_c\right)
+
O_{q,c}(n).
\]
\end{thm}

\begin{proof}
Set \(m=n-k=c\) in Lemma~\ref{lem:random-gap-bound}. We obtain
\[
\Delta_{n,k,q}
\leq
\frac{1}{1-q^{-c}}
\cdot
\frac{S_c(q)}{q-1}
\cdot
\frac{n}{c},
\]
which is exactly the claimed bound.
\end{proof}

\section{Summary and Open Problems}\label{section:Summary and Open Problems}
In this work, we studied the coverage depth of several important families of linear codes. While MDS codes are known to be optimal with respect to coverage depth whenever they exist, we showed that the \(q\)-ary Simplex and Hamming codes are also optimal for their respective parameters.

We also analyzed the coverage depth of random linear codes and showed that, in the constant-rate regime, their coverage depth approaches the MDS bound. Nevertheless, determining the coverage depth, and in particular establishing optimality, for many other important families of codes remains open. A natural direction for future work is therefore to characterize the coverage depth of additional structured code families and to identify further parameter regimes in which optimal codes can be determined. More generally, it would be interesting to obtain structural conditions that characterize, or even restrict, the class of codes that are optimal with respect to coverage depth.

\section{Acknowledgment} 
The authors would like to thank Ronny Roth and Roni Con for helpful discussions related to the analysis of random linear codes.

The research was funded by the European Union (ERC, 101045114 and EIC, DiDAX 101115134). Views and opinions expressed are, however, those of the authors only and do not necessarily reflect those of the European Union or the European Research Council Executive Agency. Neither the European Union nor the granting authority can be held responsible for them.

\bibliographystyle{IEEEtran}
\bibliography{ref}

\end{document}